\documentclass[a4paper, 12pt, oneside]{article}
\usepackage{amsmath}
\usepackage{amssymb}
\usepackage{enumerate}
\usepackage{graphicx}
\usepackage{epsf}
\usepackage{verbatim}
\usepackage{tikz}
\usepackage{systeme,mathtools}
\usetikzlibrary{matrix,arrows,decorations.pathmorphing}

\usepackage[amsmath,thmmarks]{ntheorem}
\newcounter{thm} 

\newtheorem{Thm}[thm]{Theorem}

\newtheorem{Lem}[thm]{Lemma}

\theoremstyle{nonumberplain}
\theoremheaderfont{\normalfont\itshape}
\theorembodyfont{\normalfont}
\theoremseparator{.}
\theoremsymbol{\ensuremath{\square}}
\newtheorem{proof}{Proof}

\def \C {\mathbb C}

\def \R {\mathbb R}

\def \i {\mathfrak{i}}

\def \MM {\overline{\mathbf M}}

\def \d {\mathrm{d}}
\def \tr {\mathrm{Tr}}

\def \D {\mathcal{D}}

\def \HH {\mathcal H}
\def \I {\mathcal I}

\def \M {\mathcal{M}}

\def \S {\mathcal{S}}

\def \ZZ {\mathcal{Z}}

\def \oZ {\overline{Z}}

\def \oS {\overline{S}}
\def \Wa {\widehat{\alpha}}

\def \id {\operatorname{id}}
\def \trr {\operatorname{tr}}

\begin{document}
	\begin{titlepage}
		
		\title{The identification of the extended refined open partition function and the Kontsevich-Penner matrix model} 
		\author{Gehao Wang}
		\date{}
		\maketitle

		\begin{abstract}
			The open intersection theory has been initiated by R. Pandharipande, J. P. Solomon and R. J. Tessler. In the scope of matrix model theory, A. Buryak and R. J. Tessler have constructed a matrix model $\ZZ^o$ for the open partition function based on a Kontsevich type combinatorial formula for the open intersection numbers found by R. J. Tessler. In this paper, using the Harish-Chandra-Itzykson-Zuber formula and operational calculus, we transform $\ZZ^o$ into another simple form, and define the matrix model $\ZZ_N^{o,ext,s}$ for the extended refined open partition function from it. The expression of $\ZZ_N^{o,ext,s}$ will immediately lead us to the Kontsevich-Penner matrix model $Z_N$ under the Miwa parametrization $s_i=2^ii!\trr \Lambda^{-2i-2}$. Hence it confirms the identification between the two models for general $N\geq 1$.
		\end{abstract}
		\vspace{20pt}
		
		
	\end{titlepage}

\section{Introduction}

The study of the intersection theory on the moduli space of Riemann surfaces with boundary, usually referred as the open intersection theory, was initiated by R. Pandharipande, J. P. Solomon and R. J. Tessler in \cite{PST}. They gave a complete description of the genus zero gravitational descendents, and conjectured that the generating function of the open intersection numbers satisfies the open KdV hierarchy. The descriptions of higher genus descendents was found by J. P. Solomon and R. J. Tessler in \cite{ST}. R. Dijkgraaf and E. Witten provided an interpretation of these constructions from the physical point of view in \cite{DW}. Later on, a combinatorial formula for the open intersection numbers using Feynmann diagrams was given in \cite{T}. This can be considered as a generalization of Kontsevich's combinatorial formula in the proof \cite{K} of Witten's conjecture \cite{WE}.  Recently, in some parallel theories regarding the Riemann surfaces, generating functions and integrable hierachies, the open versions have also been constructed, such as the open $r$-spin theory (\cite{BCT1},\cite{BCT2}), and open Hurwitz numbers (\cite{BTT}).

In the theory of the matrix model, it has been confirmed that the open intersection number is closely related to the following Kontsevich-Penner model
 \footnote{
 	Performing the change $H\rightarrow H-\Lambda$ on the integral will give us the Gaussian  integral form of the model. In our context, the expression \eqref{ZN} is better suited for starting our later proofs. 
 },
\begin{multline}\label{ZN}
	Z_N=c_{\Lambda,M}\exp\left(-\frac{1}{3}\trr\Lambda^3\right)\int_{\HH_{M}}[\d H]\,\exp\left( \frac{1}{6}\trr H^3\right)\exp\left( -\frac{1}{2}\trr H\Lambda^2\right)
	\frac{\det\Lambda^N}{\det(-H)^N},
\end{multline}
where $\HH_{M}$ denotes the space of Hermitian matrices of size $M$,  and, for the diagonal matrix $\Lambda=\operatorname{diag}(\lambda_1,\lambda_2,\dots,\lambda_M)$,
\begin{equation*}
	c_{\Lambda,M}=(2\pi)^{-\frac{M^2}{2}}\prod_{i=1}^M\sqrt{\lambda_i}\prod_{1\leq i< j\leq M}(\lambda_i+\lambda_j).
\end{equation*}
When $N=0$, it recovers the Kontsevich's integral. For $N=1$, A. Alexandrov proved that the matrix model $Z_1$ governs the open intersection numbers (\cite{AO1},\cite{AO2}), using Buryak's residue theorem in \cite{B2}, and it corresponds to the extended open partition function $\tau^{o,ext}$ introduced in \cite{B1} and \cite{B2} after the variable change
\begin{equation}\label{stol}
	s_i=s_i(\Lambda)=2^ii!\trr \Lambda^{-2i-2},\quad i\geq 0.
\end{equation}
In \cite{BR}, this Kontsevich-Penner model has been identified with an isomonodromic tau function, in order to give explicit expressions of the correlators.

On the other hand, using the combinatorial formula in \cite{T}, the authors in \cite{BT} constructed a matrix model $\ZZ^o$ for the open partition function $\tau^o$. This matrix model was revisited in \cite{ABT}, where the authors presented a matrix model $\ZZ_1^{o,ext,s}$ for the extended open partition function $\tau^{o,ext}$. To be specific, let
\begin{equation}\label{Dz}
	\D(z,\overline{z})=\exp\left(\frac{1}{6}z^3\right)\det\frac{H+\sqrt{\Lambda^2-\overline{z}\id_M} -z\id_M}{H+\sqrt{\Lambda^2-\overline{z}\id_M}+z\id_M},
\end{equation}
where $\id_M$ represents the identity matrix of size $M$. The model $\ZZ^o$ can be written  using the complex Gaussian integral as
\begin{multline}\label{Zo1}
	\ZZ^o=c_{\Lambda,M}\exp\left(-\frac{1}{3}\trr\Lambda^3\right)\int_{\HH_{M}}[\d H]\,\exp\left( \frac{1}{6}\trr H^3\right)\exp\left( -\frac{1}{2}\trr H\Lambda^2\right)\\
	\frac{1}{2\pi}\int_{\C}[\d  z]\,\exp\left(-\frac{1}{2}z\overline{z}\right)\D(z,\overline{z})\exp\left(\frac{1}{2}s\overline{z}\right).
\end{multline}
Here, for the complex plane $\C$ and the complex variable $z=x+y\i$, the measure is $[\d  z]=\d x\d y$. The partiton function $\tau^{o,ext}$ is uniquely determined by the equations
\begin{align}
	&\left.\tau^{o,ext}\right\vert_{s_{i}=0} = \tau^o, \quad i\geq 1.\nonumber\\
	&\frac{\partial}{\partial s_n}\tau^{o,ext}=\frac{1}{(n+1)!}\frac{\partial^{n+1}}{\partial s_0^{n+1}}\tau^{o,ext}, \quad s_0=s, \quad n\geq 0. \label{PDE}
\end{align}
Therefore, following $\ZZ^o$, we can write the matrix model $\ZZ_1^{o,ext,s}$ as
\begin{multline}\label{Z1ext}
	\ZZ_1^{o,ext,s}=c_{\Lambda,M}\exp\left(-\frac{1}{3}\trr\Lambda^3\right)\int_{\HH_{M}}[\d H]\,\exp\left( \frac{1}{6}\trr H^3\right)\exp\left( -\frac{1}{2}\trr H\Lambda^2\right)\\
	\frac{1}{2\pi}\int_{\C}[\d  z]\,\exp\left(-\frac{1}{2}z\overline{z}\right)\D(z,\overline{z})\exp\left(\sum_{i=0}^{
	\infty} \frac{2^{-i-1}}{(i+1)!}s_i\overline{z}^{i+1} \right). 
\end{multline}
Under the variable change \eqref{stol}, the model for $\tau^{o,ext}$ becomes
\begin{align*}
	\ZZ_{1}^{o,ext}
	=&c_{\Lambda,M}\exp\left(-\frac{1}{3}\trr\Lambda^3\right)\det(\Lambda)\int_{\HH_{M}}[\d H]\,\exp\left( \frac{1}{6}\trr H^3\right)\exp\left( -\frac{1}{2}\trr H\Lambda^2\right)\,I^{(C)}_1,
\end{align*}
where
\begin{equation}\label{I1C}
	I^{(C)}_1=\frac{1}{2\pi}\int_{\C}[\d  z]\,\exp\left(-\frac{1}{2}z\overline{z}\right)\,
	\frac{\D(z,\overline{z})}{\det\sqrt{\Lambda^2-\overline{z}\id_M}}.
\end{equation}
In \cite{ABT}, the authors presented a proof for $Z_1=\ZZ_1^{o,ext}$ directly through their matrix representations using the Fourier transform and the properties of the delta function. 

In this paper, inspired by the previous work in \cite{BT} and \cite{ABT}, we first present a direct proof of the identification of $Z_1$ and $\ZZ_1^{o,ext}$ using the technique of operational calculus. 
Our proof further implies that the model $\ZZ^o$ for the open partition function can be transformed into the following expression.
\begin{Thm}\label{T2}
	The matrix model $\ZZ^o$ defined in Eq.\eqref{Zo1} can be written as
	 \begin{multline}\label{Zo2}
	 	\ZZ^o=c_{\Lambda,M}\exp\left(-\frac{1}{3}\trr\Lambda^3\right)\int_{\HH_{M}}[\d H]\,\exp\left( \frac{1}{6}\trr H^3\right)\exp\left( -\frac{1}{2}\trr H\Lambda^2\right)\\
	 	\frac{1}{2\pi}\int_{\C}[\d  z]\,\exp\left(-\frac{1}{2}z\overline{z}\right)\exp\left( \frac{1}{2}s\overline{z}\right)\,\frac{\det(\sqrt{\Lambda^2-\overline{z}\id_{M}})}{\det(-H-z \id_{M})}.
	 \end{multline}
\end{Thm}

In \cite{SB}, B. Safnuk suggested an approach to define refined open intersection numbers. His combinatorial formula directly gives the Konsevich-Penner matrix model $Z_N$ for $N\geq 1$, (see also \cite{SB1}). In \cite{ABT}, the authors used a different approach to construct this refinement. In terms of the matrix model, the refinement is to switch the complex integral part of the model $\ZZ_{1}^{o,ext,s}$ in Eq.\eqref{Z1ext} to a complex matrix integral $I_N^{(C,s)}$ (see Eq.\eqref{INCs}) in the the space $\M_{N}(\C)$ of $N\times N$ complex matrices with Lebesgue measure $[\d Z]$. They conjectured that this refinement will lead to the Kontsevich-Penner model $Z_N$ under the change of variables \eqref{stol}, which gives us the equality
\begin{equation}\label{stolambda}
	\exp\left(  \sum_{i\geq 0}\frac{2^{-i-1}}{(i+1)!}s_i\trr(\oZ^t)^{i+1}\right)=\frac{\det\Lambda^N}{\det\sqrt{\Lambda^{2}\otimes \id_N-\id_M\otimes\oZ^t}}.
\end{equation}

Using Theorem \ref{T2} and the condition \eqref{PDE}, we perform the refinement on Eq.\eqref{Zo2}, and define the matrix model for the extended refined open partition function as the following.
\begin{multline}\label{ZNexts}
	\ZZ_N^{o,ext,s}=c_{\Lambda,M}\exp\left(-\frac{1}{3}\trr\Lambda^3\right)\int_{\HH_{M}}[\d H]\,\exp\left( \frac{1}{6}\trr H^3\right)\exp\left( -\frac{1}{2}\trr H\Lambda^2\right)\\
	\frac{1}{(2\pi)^{N^2}}\int_{\M_{N}(\C)}[\d Z]\,\exp\left( -\frac{1}{2}\trr(Z\oZ^t)\right)\exp\left( \frac{1}{6}\trr(Z^3)\right)\\
	\frac{\det\sqrt{\Lambda^2\otimes \id_{N}-\id_M\otimes \oZ^t}}{\det(-H\otimes\id_N-\id_M\otimes Z)}
	\exp\left(  \sum_{i\geq 0}\frac{2^{-i-1}}{(i+1)!}s_i\trr(\oZ^t)^{i+1}\right).
\end{multline}
Furthermore, we define
\begin{equation}\label{ZNext}
	\left.\ZZ_N^{o,ext}= \ZZ_N^{o,ext,s}\right\vert_{s_i=s_i(\Lambda)}.
\end{equation}
Then, after the substitution \eqref{stolambda} in Eq.\eqref{ZNexts}, there is an obvious cancellation, and the model $\ZZ_N^{o,ext}$ immediately implies that:
\begin{Thm}\label{T1}
	For the Kontsevich-Penner model $Z_N$ and the model $\ZZ_N^{o,ext}$ with $N\geq 1$, we have
	 $$Z_N=\ZZ_N^{o,ext}.$$
\end{Thm}
Our main goal of this paper is to give the above identification between the model  $\ZZ_N^{o,ext}$ for the extended refined open partition function and the Kontsevich-Penner matrix model for the general $N\geq 1$.

The paper is organized as follows. In Sect.\ref{S2}, we prove the case $N=1$ of Theorem \ref{T1}. In Sect.\ref{S3}, we prove Theorem \ref{T2} and discuss the Virasoro constraints for $\tau^o$. In Sect.\ref{S4}, we address some difficulties on relating  the conjectural complex matrix model in \cite{ABT} to $Z_N$ for $N\geq 2$ using our approach.

\section{The extended open partition function}\label{S2}

For the diagonal matrix $\Lambda$, let 
\begin{equation*}
	\Delta_{M}(\Lambda)=\prod_{1\leq i,j\leq M}(\lambda_j-\lambda_i).
\end{equation*}
Let $f(H)$ be a unitary invariant function. By the Harish-Chandra-Itzykson-Zuber (HCIZ) formula (\cite{HC},\cite{IZ}), we can express the following Hermitian Gaussian integral as
\begin{multline*}
	\int_{\HH_{M}} [\d H]\,f(H)\exp\left(-\frac{1}{2}\trr H^2\Lambda\right)
	\\
	=\frac{(2\pi)^{\frac{M^2-M}{2}}}{M!\Delta_{M}(\Lambda)}\int_D \d \MM\, f(\MM)\det(e^{-\frac{1}{2}m_i^2\lambda_j})_{1\leq i,j\leq M}\prod_{1\leq i<j\leq M}\frac{m_j-m_i}{m_j+m_i},
\end{multline*}
where $\MM=\text{diag}(m_1,m_2,\dots,m_{M})$ and $\d\MM=\d m_1\d m_2\dots \d m_{M}$. 

For $s\in\R$ and $t>0$, we consider the following two Hermitian matrix integrals,
\begin{align}
	I^{(t)}_{M,y}=&\int_{\HH_{M}}[\d H]\,\exp\left( \frac{1}{6}\trr H^3\right)\exp\left( -\frac{1}{2}\trr H^2\Lambda\right)\nonumber\\
	&\quad \int_{-\infty}^{\infty}\d y\,	\exp\left( -\frac{(y+\trr H-s)^2}{4t}\right) e^{\frac{y^3}{6}}\det\frac{H-y\id_M}{H+y\id_M}\label{ave1},
\end{align}
and
\begin{equation*}
	I^{(t)}_{M+1}=\int_{\HH_{M+1}}[\d H']\,\exp\left( \frac{1}{6}\trr (H')^3\right)\exp\left( -\frac{1}{2}\trr (H')^2\Lambda'\right)\exp\left( -\frac{(\trr H'-s)^2}{4t}\right),
\end{equation*}
where $H'$ is the Hermitian matrix with size $M+1$ and $\Lambda'=\text{diag}(\lambda_1,\lambda_2,\dots,\lambda_M,\lambda_{M+1})$ with $\lambda_{M+1}=0$. We first use the same polynomial average technique as the one used in \cite{BT} to prove the following lemma.  
\begin{Lem}\label{ytoM+1}
	\begin{equation*}
		I^{(t)}_{M,y}
		=\frac{\det(\Lambda)}{(2\pi)^M} I^{(t)}_{M+1}.
	\end{equation*}
\end{Lem}
\begin{proof}
Using the HCIZ formula, we can express the Hermitian matrix integral \eqref{ave1} as
\begin{align*}
	I^{(t)}_{M,y}&=\frac{(2\pi)^{\frac{M^2-M}{2}}}{M!\Delta_{M}(\Lambda)}\int_{-\infty}^{\infty}\d y\,\int_{\R^M}\d \MM\,\exp\left( \frac{1}{6}\trr \MM^3\right)e^{\frac{y^3}{6}}\exp\left( -\frac{(y+\trr \MM-s)^2}{4t}\right)\\
	&\quad\quad\quad\det\left(e^{\frac{1}{2}m_i^2\lambda_j}\right)_{1\leq i,j\leq M}\det\frac{\MM-y\id_M}{\MM+y\id_M}\prod_{1\leq i<j\leq M}\frac{m_j-m_i}{m_j+m_i}.
\end{align*}
On the other hand, for the diagonal matrix $\MM'$ of size $M+1$, using the HCIZ formula again, we can transform the integral $I^{(t)}_{M+1}$ into
\begin{align}
	I^{(t)}_{M+1}
	=&\frac{(2\pi)^{\frac{M^2+M}{2}}}{(M+1)!\Delta_{M}(\Lambda)}\frac{(-1)^M}{\det(\Lambda)}\int_{\R^{M+1}}\d \MM\,\exp\left( \frac{1}{6}\trr \MM'^3\right)\exp\left( -\frac{(\trr \MM'-s)^2}{4t}\right)\nonumber\\
	&\quad\quad\quad\det\left(e^{-\frac{1}{2}m_i^2\lambda_j}\right)_{1\leq i,j\leq M+1}\prod_{1\leq i<j\leq M+1}\frac{m_j-m_i}{m_j+m_i}.\label{Int:M+1}
\end{align}
Since $\lambda_{M+1}=0$, the determinant part in the above integral can be re-written as the summation of $M+1$ terms as:
\begin{equation*}
	\det\left(e^{-\frac{1}{2}m_i^2\lambda_j}\right)_{1\leq i,j\leq M+1}=\sum_{k=1}^{M+1}(-1)^{M+1-k} \det\left(e^{-\frac{1}{2}m_i^2\lambda_j}\right)_{\substack{1\leq i\leq M+1,i\neq k\\1\leq j\leq M}}.
\end{equation*}
For $1\leq k\leq M+1$, the product part in the integral $I^{(t)}_{M+1}$ can be re-written as
\begin{align}
	&\prod_{\substack{1\leq i<j\leq M+1\\i,j\neq k}}\frac{m_j-m_i}{m_j+m_i}\prod_{k<j\leq M+1}\frac{m_j-m_k}{m_j+m_k}\prod_{1\leq i<k}\frac{m_k-m_i}{m_k+m_i}\nonumber\\
	=&(-1)^{k-1}\prod_{\substack{1\leq i<j\leq M+1\\i,j\neq k}}\frac{m_j-m_i}{m_j+m_i}\prod_{\substack{1\leq j\leq M+1\\j\neq k}}\frac{m_j-m_k}{m_j+m_k}.\label{prod}
\end{align} 
We can treat the variable $m_k=y$ as an additional variable, and transform Eq.\eqref{Int:M+1} into 
\begin{align*}
	I^{(t)}_{M+1}=&\frac{(2\pi)^{\frac{M^2+M}{2}}}{M!\Delta_{M}(\Lambda)\det(\Lambda)}\int_{\R^M}\d \MM\,\exp\left( \frac{1}{6}\trr \MM^3\right) \det\left(e^{-\frac{1}{2}m_i^2\lambda_j}\right)_{1\leq i,j\leq M}\prod_{1\leq i<j\leq M }\frac{m_j-m_i}{m_j+m_i}\nonumber\\
	&\quad\quad\quad\int_{-\infty}^{\infty}\d y\,e^{\frac{y^3}{6}}\exp\left( -\frac{(y+\trr \MM-s)^2}{4t}\right)\det\frac{\MM-y\id_M}{\MM+y\id_M}\\
	=&\frac{(2\pi)^M}{\det(\Lambda)}  \int_{\HH_M}[\d H]  \exp\left( \frac{1}{6}\trr H^3 \right) \exp\left( -\frac{1}{2}\trr H^2\Lambda\right)\\
	&\quad\quad\int_{-\infty}^{\infty} \d y \, e^{\frac{y^3}{6}}\exp\left( -\frac{(y+\trr H-s)^2}{4t}\right) \det\frac{H-y\id_M}{H+y\id_M}.
\end{align*}
This proves the lemma.
\end{proof}

Next, we prove the case $N=1$ of Theorem \ref{T1}.

\begin{proof}[of $Z_1=\ZZ_1^{o,ext}$]
Let us define the series $I^{(C)}_1(s)$ to be 
\begin{equation}\label{I1Cs}
	I^{(C)}_1(s)=\frac{1}{2\pi}\int_{\C}[\d  z]\,\exp\left(-\frac{1}{2}z\overline{z}\right)\,\exp\left( \frac{1}{2}s\overline{z}\right)F(z,\overline{z}),
\end{equation}
where, compared with Eq.\eqref{Dz},
\begin{equation}\label{F1}
	F(z,\overline{z})=\frac{\D(z,\overline{z})}{\det\sqrt{\Lambda^2-\overline{z}\id_M}}.
\end{equation}
Note that, when $s=0$, we obtain exactly $I^{(C)}_1$ in Eq.\eqref{I1C}. By Eq.\eqref{equal}, we have
\begin{equation*}
	I^{(C)}_1(s)=\left.\left(\exp\left(2\frac{\partial^2}{\partial s\partial s_{-}}\right)\cdot F(s,s_{-})\right)\right\vert_{s_{-}=0}.
\end{equation*}
Now, observe that
\begin{equation}\label{trLambdaop}
	\left.\left(\frac{\partial}{\partial s_{-}}\cdot\trr\sqrt{\Lambda^2-s_{-}\id_M}\right)\right\vert_{s_{-}=0}=-\frac{1}{2}\trr \Lambda^{-1}\frac{\partial}{\partial \Lambda}\cdot \Lambda,
\end{equation}
where
\begin{equation*}
	\trr \Lambda^{-1}\frac{\partial}{\partial \Lambda}=\sum_{i=1}^M \frac{1}{\lambda_i}\frac{\partial}{\partial \lambda_i}.
\end{equation*}
In fact, for functions formed by 
\begin{equation*}
	\trr\left( \Lambda^2-s_{-}\id_M\right)^{\frac{k}{2}}=\sum_{i=1}^M\left( \lambda_i^2-s_{-}\right)^{\frac{k}{2}}, 
\end{equation*} 
we have
\begin{align*}
	\frac{\partial}{\partial s_{-}}\cdot\trr\left( \Lambda^2-s_{-}\id_M\right)^{\frac{k}{2}}
	=&-\frac{k}{2}\trr\left( \Lambda^2-s_{-}\id_M\right)^{\frac{k}{2}-1}\\
	=&-\frac{1}{2}\trr \Lambda^{-1}\frac{\partial}{\partial \Lambda}\cdot\trr\left( \Lambda^2-s_{-}\id_M\right)^{\frac{k}{2}}.
\end{align*}
Therefore, after replacing $\sqrt{\Lambda^2-s_{-}\id_M}$ by $\Lambda$ in $F(s,s_{-})$, we have
\footnote
{
This explains why we perform the change $H\rightarrow H+\Lambda$ on the Hermitian Gaussian integral and use the expression of $\D(z,\overline{z})$ instead.
}
\begin{equation}\label{I1s}
	I^{(C)}_1(s)=\exp\left(-\trr \Lambda^{-1}\frac{\partial}{\partial \Lambda}\frac{\partial}{\partial s}\right)\cdot e^{\frac{s^3}{6}}\frac{1}{\det\Lambda}\det\frac{H+\Lambda -s\id_M}{H+\Lambda+s\id_M},
\end{equation}
If we define
\begin{equation*}
	\I_M^{(ext,s)}=\exp\left(-\frac{1}{3}\trr\Lambda^3\right)\int_{\HH_{M}}[\d H]\,\exp\left( \frac{1}{6}\trr H^3\right)\exp\left( -\frac{1}{2}\trr H\Lambda^2\right)\,I^{(C)}_1(s),
\end{equation*}
then,  
\begin{equation}\label{Z1o}
		\ZZ_{1}^{o,ext}=c_{\Lambda,M}\det(\Lambda)\left.\I_M^{(ext,s)}\right\vert_{s=0}.
\end{equation}

We focus on the transformation of the integral $\I_M^{(ext,s)}$. First, using the conjugation
\begin{multline*}
	\exp\left(\trr \Lambda^{-1}\frac{\partial}{\partial \Lambda}\frac{\partial}{\partial s}\right)\exp\left( -\frac{1}{2}\trr H\Lambda^2\right)\exp\left(-\trr \Lambda^{-1}\frac{\partial}{\partial \Lambda}\frac{\partial}{\partial s}\right)\\
	=\exp\left( -\frac{1}{2}\trr H\Lambda^2\right)\exp\left(-\trr H\frac{\partial}{\partial s}\right),
\end{multline*}
we can move the differential operator 
\begin{equation}\label{do}
	\exp\left(-\trr \Lambda^{-1}\frac{\partial}{\partial \Lambda}\frac{\partial}{\partial s}\right)
\end{equation}
out of the integral in $\I_M^{(ext,s)}$ as:
\begin{equation}\label{ZMs}
	\I_M^{(ext,s)}=\exp\left(-\frac{1}{3}\trr\Lambda^3\right)\exp\left(-\trr \Lambda^{-1}\frac{\partial}{\partial \Lambda}\frac{\partial}{\partial s}\right)\cdot \I_M^{(1)},
\end{equation}
where
\begin{multline*}
	\I_M^{(1)}=\int_{\HH_{M}}[\d H]\,\exp\left( \frac{1}{6}\trr H^3\right)\exp\left( -\frac{1}{2}\trr H\Lambda^2\right)\\
	\exp\left(-\trr H\frac{\partial}{\partial s}\right)\cdot e^{\frac{s^3}{6}}\frac{1}{\det\Lambda}\det\frac{H+\Lambda -s\id_M}{H+\Lambda+s\id_M}.
\end{multline*}
Next, we perform the change of variables $H\rightarrow H-\Lambda$ on the integral $\I_M^{(1)}$. 
This gives us
\begin{equation}\label{ZM1}
	\I_M^{(1)}=\exp\left(\trr\Lambda\frac{\partial}{\partial s}\right)\exp\left(\frac{1}{3}\trr\Lambda^3\right)\I_M^{(2)},
\end{equation}
where $\I_M^{(2)}$ now becomes a Gaussian integral as
\begin{multline*}
	\I_M^{(2)}=\frac{1}{\det\Lambda}\int_{\HH_{M}}[\d H]\,\exp\left( \frac{1}{6}\trr H^3\right)\exp\left( -\frac{1}{2}\trr H^2\Lambda\right)\\
	\exp\left(-\trr H\frac{\partial}{\partial s}\right)\cdot e^{\frac{s^3}{6}}\det\frac{H-s\id_M}{H+s\id_M}.
\end{multline*}
This expression satisfies the unitary invariant condition for the HCIZ formula. Let us perform the Weierstrass transform on the integral $\I_M^{(2)}$, (see Eq.\eqref{WT} in Appendix A.2 for some details). By Eq.\eqref{ave1},
 \begin{equation*}
	\exp\left( t\frac{\partial^2}{\partial s^2}\right)\cdot \I_M^{(2)}=\frac{1}{\det(\Lambda)}\frac{1}{2\sqrt{\pi t}}I^{(t)}_{M,y}.
\end{equation*}
Using Lemma \ref{ytoM+1}, we have 
\begin{equation}\label{polyave1}
 	\exp\left( t\frac{\partial^2}{\partial s^2}\right)\cdot \I_M^{(2)}
 	=\frac{1}{(2\pi)^M}\frac{1}{2\sqrt{\pi t}}I^{(t)}_{M+1}.
 \end{equation}
 
Suppose the Hermitian matrix $H'$ is in the form
 \begin{equation*}
 	H'=\begin{bmatrix}
 		H & C\\
 		\bar{C} ^t & y
 	\end{bmatrix},
 	\mbox{ where }
 	\bar{C} ^t= [\bar{C}_1,\bar{C}_2,\dots,\bar{C}_M], \quad C_i\in\C. 
 \end{equation*}
 Then, using the following expressions
 \begin{align*}
 	&\trr H'^3=\trr H^3+3\trr (\bar{C}^t H C)+y^3+3y C \bar{C} ^t;\\
 	&\trr (H')^2\Lambda'= \trr H^2\Lambda+ \trr (\bar{C}^t \Lambda C);\\
 	&\trr H'=\trr H+y,
 \end{align*}
 and the complex Gaussian integral
 \begin{equation*}
 	\frac{1}{(2\pi)^M}\int_{\C^M}\exp\left(-\frac{1}{2}\trr(\bar{C}^t A C)\right)\prod_{i=1}^M [\d C_i]=\frac{1}{\det(A)},
 \end{equation*}
 we can transform the Hermitian integral \eqref{polyave1} into
 \begin{align*}
 	 \exp\left( t\frac{\partial^2}{\partial s^2}\right)\cdot \I_M^{(2)}&=\int_{\HH_{M}}[\d H]\,\exp\left( \frac{1}{6}\trr H^3\right)\exp\left( -\frac{1}{2}\trr H^2\Lambda\right)\\
 	&\quad\quad\quad\quad\frac{1}{2\sqrt{\pi t}}\int_{-\infty}^{\infty}\d y\,\exp\left( -\frac{(\trr H+y-s)^2}{4t}\right)\frac{1}{\det(\Lambda-H-y \id_{M})}\\
 	=&\exp\left(-\frac{1}{3}\trr\Lambda^3\right)\int_{\HH_{M}}[\d H]\,\exp\left( \frac{1}{6}\trr H^3\right)\exp\left( -\frac{1}{2}\trr H\Lambda^2\right)\\
 	&\quad\quad\quad\quad\frac{1}{2\sqrt{\pi t}}\int_{-\infty}^{\infty}\d y\,\exp\left( -\frac{(\trr H+\trr\Lambda+y-s)^2}{4t}\right)\frac{1}{\det(-H-y \id_{M})}.
 \end{align*}
Taking into account Eq.\eqref{ZMs} and Eq.\eqref{ZM1}, we have
 \begin{align}
 	&\exp\left( t\frac{\partial^2}{\partial s^2}\right)\cdot\I_M^{(ext,s)}\nonumber\\
 	=&\exp\left(-\frac{1}{3}\trr\Lambda^3\right)\int_{\HH_{M}}[\d H]\,\exp\left( \frac{1}{6}\trr H^3\right)\exp\left( -\frac{1}{2}\trr H\Lambda^2\right)\nonumber\\
 	&\quad\quad\exp\left(-\trr \Lambda^{-1}\frac{\partial}{\partial \Lambda}\frac{\partial}{\partial s}\right)\frac{1}{2\sqrt{\pi t}}\int_{-\infty}^{\infty}\d y\,\exp\left( -\frac{(y-s)^2}{4t}\right)\frac{1}{\det(-H-y \id_{M})}\nonumber\\
 	=&\exp\left(-\frac{1}{3}\trr\Lambda^3\right)\int_{\HH_{M}}[\d H]\,\exp\left( \frac{1}{6}\trr H^3\right)\exp\left( -\frac{1}{2}\trr H\Lambda^2\right)\nonumber\\
 	&\quad\quad\quad\quad\quad\quad\exp\left(-\trr \Lambda^{-1}\frac{\partial}{\partial \Lambda}\frac{\partial}{\partial s}\right)\exp\left( t\frac{\partial^2}{\partial s^2}\right)\cdot\frac{1}{\det(-H-s \id_{M})}\label{pa}.
 \end{align}
Here we move the differential operator \eqref{do} back inside the integral $\I_M^{(ext,s)}$ again. The action of this operator on $\det(-H-s \id_{M})$ is trivial. Also, if the integral
\begin{equation*}
	\int_{-\infty}^{\infty}\exp\left( -\frac{(y-s)^2}{4t}\right)F(y) \,\d y=\int_{-\infty}^{\infty}\exp\left( -\frac{y^2}{4t}\right)F(y+s) \,\d y
\end{equation*}
is zero, then the series $F(y)$ is $0$. In other words, we should have
 \begin{multline}\label{IMe}
 	\I_M^{(ext,s)}=\exp\left(-\frac{1}{3}\trr\Lambda^3\right)\int_{\HH_{M}}[\d H]\,\exp\left( \frac{1}{6}\trr H^3\right)\exp\left( -\frac{1}{2}\trr H\Lambda^2\right)
 	\frac{1}{\det(-H-s \id_{M})}.
 \end{multline}
Finally, by Eq.\eqref{Z1o}, we have $\ZZ_{1}^{o,ext}=Z_1$. This completes the proof.
\end{proof}

\section{The matrix model for the open partition function}\label{S3}
In this section, we first show that how to deduce the expression $\ZZ^o$ in Eq.\eqref{Zo2}. Then, we show that this expression will allow us to describe the Virasoro constraints like the way in \cite{AO2}.

\begin{proof}[of Theorem \ref{T2}]
Let
\begin{align*}
	I^o(s)&=\frac{1}{2\pi}\int_{\C}[\d  z]\,\exp\left(-\frac{1}{2}z\overline{z}\right)\,\exp\left( \frac{1}{2}s\overline{z}\right)\D(z,\overline{z}),
\end{align*}
such that
\begin{equation}\label{Zo3}
	\ZZ^o=c_{\Lambda,M}\exp\left(-\frac{1}{3}\trr\Lambda^3\right)\int_{\HH_{M}}[\d H]\,\exp\left( \frac{1}{6}\trr H^3\right)\exp\left( -\frac{1}{2}\trr H\Lambda^2\right)\,I^o(s).
\end{equation}
Then, by Eq.\eqref{equal},
\begin{equation*}
	I^o(s)=\left.\left(\exp\left(2\frac{\partial^2}{\partial s\partial s_{-}}\right)\cdot \D(s,s_{-})\right)\right\vert_{s_{-}=0}.
\end{equation*}
From Eq.\eqref{F1}, we can see that
\begin{equation*}
	\D(s,s_{-})=\det\sqrt{\Lambda^2-s_{-}\id_{M}}\,F(s,s_{-}).
\end{equation*}
Due to this difference, after comparing with Eq.\eqref{pa}, we can see that the Hermitian integral \eqref{Zo3} can be expressed as
\begin{multline}\label{Zo4}
	\ZZ^o=c_{\Lambda,M}\exp\left(-\frac{1}{3}\trr\Lambda^3\right)\int_{\HH_{M}}[\d H]\,\exp\left( \frac{1}{6}\trr H^3\right)\exp\left( -\frac{1}{2}\trr H\Lambda^2\right)\\
	\exp\left(-\trr \Lambda^{-1}\frac{\partial}{\partial \Lambda}\frac{\partial}{\partial s}\right)\cdot\frac{\det(\Lambda)}{\det(-H-s \id_{M})}.
\end{multline}
By Eq.\eqref{equal} and Eq.\eqref{trLambdaop}, we can restore the complex integral as
\begin{multline*}
\exp\left(-\trr \Lambda^{-1}\frac{\partial}{\partial \Lambda}\frac{\partial}{\partial s}\right)\cdot\frac{\det(\Lambda)}{\det(-H-s \id_{M})}\\
	=\frac{1}{2\pi}\int_{\C}[\d  z]\,\exp\left( \frac{1}{2}s\overline{z}\right)\exp\left(-\frac{1}{2}z\overline{z}\right)\,\frac{\det(\sqrt{\Lambda^2-\overline{z}\id_{M}})}{\det(-H-z \id_{M})}.
\end{multline*}
\end{proof}

\noindent
{\bf Remark}: In \cite{BT}, the authors added the imaginary number $\i$ to the matrix integral as
\begin{equation*}
	\tau^o=\left.\left\{\exp\left(2\frac{\partial^2}{\partial s\partial s_{-}}\right)\cdot\left(e^{\frac{s^3}{6}}f_M^o\right)\right\}\right\vert_{s_{-}=0},
\end{equation*}
where
\begin{align*}
	f_M^o=&\int_{\HH_{M}}[\d H]\,\exp\left( \frac{\i}{6}\trr H^3\right)\exp\left( -\frac{1}{2}\trr H^2\Lambda\right)\\
	&\quad\quad\det\frac{\Lambda+\sqrt{\Lambda^2-s_{-}\id_M}-\i H+s\id_M}{\Lambda+\sqrt{\Lambda^2-s_{-}\id_M}-\i H-s\id_M},
\end{align*}
and they used the following integral transform 
\footnote[1]{
This integral transform on $f$ is (see, e.g., Eq.\eqref{invW})
\begin{equation*}
	\Phi_s^{form}[f](z)=\left.\left\{\exp\left(-\frac{1}{2t}\frac{\partial^2}{\partial  z^2}\right)\exp\left(\frac{1}{6}(z-t)^3-\frac{1}{6}z^3\right)\cdot f(z) \right\}\right\vert_{t=z}.
\end{equation*}
}
on power series $f=\exp\left(\frac{s^3}{6}\right)\widetilde{f}$,
\begin{align*}
	\Phi_s^{form}[f](z)&:=\sqrt{\frac{z}{2\pi}}\int_{\R} \widetilde{f}(-\i s+z)e^{\frac{\i s^3}{6}}e^{-\frac{1}{2}s^2z} \d s.
\end{align*}
They obtained a simplied expression of $\Phi_s^{form}[\tau^o]$ using Lemma 3.2 in \cite{BT}. The proof this Lemma requires the insertion of $\i$ and $z$ being an arbitrary positive real number, (so that after the integration by parts, the non-integral term is zero). 
\begin{flushright}
	$\Box$
\end{flushright}

\subsection{Virasoro constraints for $\tau^o$}
We consider the Heisenberg operators $\widehat{\alpha}_n$ that correspond to $-z^n$ in the $W_{1+\infty}$ algebra:
\[
\widehat{\alpha}_n=
\begin{cases}
	q_{-n} & n<0 \\
	n\frac{\partial}{\partial q_n} &  n>0.
\end{cases}
\]
The Virasoro operators $\widehat{L}_n$ correspond to $-z^{1+n}\partial_z-\frac{n+1}{2}z^n$. We can write the operators $\widehat{L}_{-2m-2}$ and $\widehat{L}_{2m}$ as
\begin{align*}
	\widehat{L}_{-2m-2}&=\sum_{i>0}q_{i+2m+2}\Wa_i+\frac{1}{2}\sum_{0<i<2m-2}q_iq_{2m+2-i},\\
	\widehat{L}_{2m}&=\sum_{j>0}q_j\Wa_{j+2m}+\frac{1}{2}\sum_{0<j<2m}\Wa_{j}\Wa_{2m-j},
\end{align*} 
Using our notations, the Virasoro constraints for $\tau^o$ (see, e.g., \cite{PST}, \cite{B1}, \cite{B2}, \cite{BT}) can be expressed as the following
\begin{align*}
	& \left( \widehat{L}_{-2}-\frac{\partial}{\partial q_{1}}+s\right)\cdot\tau^o=0,\\
	&	\left( \frac{1}{2^{n+1}}\widehat{L}_{2n}-\frac{1}{2^{n+1}}\Wa_{2n+3}+\frac{\partial^{n+1}}{\partial s^{n+1}}s-\frac{n+1}{4}\frac{\partial^{n}}{\partial s^{n}}\right)\cdot\tau^o=0,\quad n\geq 0.
\end{align*}
Furthermore, the tau-function is independent on even times, 
$$\frac{\partial}{\partial q_{2k}} \tau^o=0, \quad k\geq 1.$$

Let us consider the following tau-function:
\begin{equation*}
	\widetilde{\tau}=\left.\widetilde{Z}\right\vert_{q_{k}=\tr \Lambda^{-k}}, \quad \widetilde{Z}=c_{\Lambda,M}\det(\Lambda)\I_M^{(ext,s)},
\end{equation*}
where $\I_M^{(ext,s)}$ is defined in Eq.\eqref{IMe}. From the model Eq.\eqref{Zo4} for the tau-function $\tau^o$, since
\begin{align*}
	&\exp\left(-\trr \Lambda^{-1}\frac{\partial}{\partial \Lambda}\frac{\partial}{\partial s}\right)\cdot\frac{\det(\Lambda)}{\det(-H-s \id_{M})}\\
	=&\det\sqrt{\Lambda^2-2\partial_s}\cdot\frac{1}{\det(-H-s \id_{M})}\\
	=&\exp\left(\frac{1}{2}\trr\log\left(\id_M-2\Lambda^{-2}\partial_s\right)\right)\cdot\frac{\det(\Lambda)}{\det(-H-s \id_{M})},
\end{align*}
we can see that
\begin{equation}\label{oS}
	\tau^o=e^{-S}\cdot\widetilde{\tau},
\end{equation}
where, under the Miwa parametrization $q_{k}=\tr \Lambda^{-k}$,
\begin{equation*}
	-S=\exp\left(\frac{1}{2}\trr\log\left(\id_M-2\Lambda^{-2}\partial s\right)\right)=-\sum_{k=1}^{\infty}2^k\frac{q_{2k}}{2k}\frac{\partial^k}{\partial s^k}.
\end{equation*}
Using the commutator relation $[q_n,\widehat{L}_k]=-n\widehat{\alpha}_{k-n}$, we can compute the Virasoro constraints for $\widetilde{\tau}$ as
\begin{align*}
	&e^S\left( \widehat{L}_{-2}-\frac{\partial}{\partial q_{1}}+s\right)e^{-S}\\
	=&\widehat{L}_{-2}-\frac{\partial}{\partial q_{1}}-\sum_{k=1}^{\infty}2^kq_{2k+2}\frac{\partial^k}{\partial s^k}+s+\sum_{k=1}^{\infty}2^{k-1}q_{2k}\frac{\partial^{k-1}}{\partial s^{k-1}}\\
	=&\widehat{L}_{-2}-\frac{\partial}{\partial q_{1}}+s+q_2,
\end{align*}
and
\begin{align*}
	&e^S	\left( \frac{1}{2^{n+1}}\widehat{L}_{2n}-\frac{1}{2^{n+1}}\widehat{\alpha}_{2n+3}+\frac{\partial^{n+1}}{\partial s^{n+1}}s-\frac{n+1}{4}\frac{\partial^{n}}{\partial s^{n}}\right)e^{-S}\\
	=&\frac{1}{2^{n+1}}\left(\widehat{L}_{2n}-\sum_{k=1}^{n-1}2^{k}\widehat{\alpha}_{2n-2k}\frac{\partial^k}{\partial s^k}- \sum_{k=n+1}^{\infty}2^{k}q_{2k-2n}\frac{\partial^k}{\partial s^k} \right)\\
	&\quad\quad-\frac{1}{2^{n+1}}\widehat{\alpha}_{2n+3}+\frac{n-1}{4}\frac{\partial^{n}}{\partial s^{n}}+\frac{\partial^{n+1}}{\partial s^{n+1}}s+\sum_{k=1}^{\infty}2^{k-1}q_{2k}\frac{\partial^{n+k}}{\partial s^{n+k}}-\frac{n+1}{4}\frac{\partial^{n}}{\partial s^{n}}\\
	=&\frac{1}{2^{n+1}}\widehat{L}_{2n}-\frac{1}{2^{n+1}}\widehat{\alpha}_{2n+3}-\frac{1}{2^{n+1}}\sum_{k=1}^{n-1}2^{k}\widehat{\alpha}_{2n-2k}\frac{\partial^k}{\partial s^k}+\frac{\partial^{n+1}}{\partial s^{n+1}}s-\frac{1}{2}\frac{\partial^{n}}{\partial s^{n}}.
\end{align*}
In fact, under the parametrization
\begin{equation*}
	s=\trr\Lambda^{-2}=q_2, \quad t_k=\frac{1}{k}\trr\Lambda^{-k}=\frac{1}{k}q_k,
\end{equation*} 
the relation Eq.\eqref{oS} and the difference between the Virasoro constraints of $\tau^o$ and $\widetilde{\tau}$ agree with the results presented in Sect. 7 of \cite{AO2}.

\section{Further remarks on the extened refined open partition function}\label{S4}
 The expression of $I_N^{(C,s)}$ mentioned in the introduction is 
\begin{multline}\label{INCs}
	I_N^{(C,s)}=\frac{1}{(2\pi)^{N^2}}\int_{\M_{N}(\C)}[\d Z]\,\exp\left( -\frac{1}{2}\trr(Z\oZ^t)\right)\exp\left( \frac{1}{6}\trr(Z^3)\right)\\
	\det\frac{H\otimes\id_N+\sqrt{\Lambda^2\otimes \id_{N}-\id_M\otimes \oZ^t}-\id_M\otimes Z}{H\otimes\id_N+\sqrt{\Lambda^2\otimes \id_{N}-\id_M\otimes \oZ^t}+\id_M\otimes Z}\\
	\exp\left(  \sum_{i\geq 0}\frac{2^{-i-1}}{(i+1)!}s_i\trr(\oZ^t)^{i+1}\right).
\end{multline}
The conjectural matrix model for the extended refined open partition function introduced in \cite{ABT} is \begin{equation*}
	\tau_N^{o,ext}=c_{\Lambda,M}\left(-\frac{1}{3}\trr\Lambda^3\right)\det\Lambda^N\int_{\HH_{M}}[\d H]\,\exp\left( \frac{1}{6}\trr H^3\right)\exp\left( -\frac{1}{2}\trr H\Lambda^2\right)I_N^{(C)},
\end{equation*}
where $I_N^{(C)}$ is the complex matrix integral
\begin{equation}\label{INC}
	I_N^{(C)}=\frac{1}{(2\pi)^{N^2}}\int_{\M_{N}(\C)}[\d Z]\,\exp\left( -\frac{1}{2}\trr(Z\oZ^t)\right)F_N(Z,\oZ^t),
\end{equation}
and
\begin{multline}\label{FNZ}
	F_N(Z,\oZ^t)=\exp\left( \frac{1}{6}\trr(Z^3)\right)
	\det\frac{H\otimes\id_N+\sqrt{\Lambda^2\otimes \id_{N}-\id_M\otimes \oZ^t}-\id_M\otimes Z}{H\otimes\id_N+\sqrt{\Lambda^2\otimes \id_{N}-\id_M\otimes \oZ^t}+\id_M\otimes Z}\\
	\frac{1}{\det\sqrt{\Lambda^{2}\otimes \id_N-\id_M\otimes\oZ^t}}.
\end{multline}
Let the Schur decomposition of $Z$ be $Z=USU^{-1}$, where $U$ is unitary matrix, and $S$ is a complex upper triangular matrix. The complex matrix integral $I_N^{(C)}$ is invariant under the Schur decomposition, and therefore it can be reduced to
\begin{equation}\label{INC2}
	I_N^{(C)}=\frac{1}{(2\pi)^{\frac{N^2+N}{2}}}\int_{\S_{N}(\C)}[\d S]\,\exp\left( -\frac{1}{2}\trr(S\oS^t)\right)F_N(S,\oS^t),
\end{equation}
where $\S_{N}(\C)$ is the set of complex upper triangular matrix. 

For this matrix model, the technique in our previous proofs doesn't seem to work when $N\geq 2$. Let us consider the simplest case $M=1$ and $N=2$. Suppose the triangular complex matrices $S$ and $\oS^t$ are in the form
\begin{equation*}
	S=\begin{bmatrix}
		z_1 & s\\
		0 & z_2
	\end{bmatrix} \quad \mbox{and}\quad 
	\oS^t=\begin{bmatrix}
		\bar{z}_1 & 0\\
		\bar{s} & \bar{z}_2
	\end{bmatrix} .
\end{equation*}
The integral becomes
\begin{equation*}
	\int_{\R}\d h\,\exp\left( \frac{1}{6}h^3\right)\exp\left( -\frac{1}{2}h\lambda^2\right)I_2^{(C)},
\end{equation*}
where
\begin{equation}
	I_2^{(C)}=\frac{1}{(2\pi)^{3}}\int_{\C^3}[\d z_1]\,[\d z_2]\,[\d s]\,\exp\left( -\frac{1}{2}(z_1\bar{z}_1+z_2\bar{z}_2+s\bar{s})\right)F_2(S,\overline{S}^t),
\end{equation}
where
\begin{equation}
	F_2(S,\overline{S})=\exp\left( \frac{1}{6}\trr(S^3)\right)
	\det\frac{h\id_2+\sqrt{\lambda^2 \id_{2}- \overline{S}^t}- S}{h\id_2+\sqrt{\lambda^2 \id_{2}-\overline{S}^t}+ S}
	\frac{1}{\det\sqrt{\lambda^{2} \id_2-\overline{S}^t}}.
\end{equation}
By a straightforward computation, we have
\begin{align*}
	(\oS^t)^n&=\begin{bmatrix}
		\bar{z}_1^n & 0\\
		\bar{s}\frac{\bar{z}_1^n-\bar{z}_2^n}{\bar{z}_1-\bar{z}_2} & \bar{z}_2^n
	\end{bmatrix},\\
	\sqrt{\id_2-\lambda^{-2}\oS^t}&=\sum_{k=0}^{\infty}(-1)^k\lambda^{-2k}\binom{\frac{1}{2}}{k}(\oS^t)^k\\
	&=\begin{bmatrix}
		\sqrt{1-\lambda^{-2}\bar{z}_1} & 0\\
		-\frac{\lambda^{-1}\bar{s}}{\sqrt{\lambda^{2}-\bar{z}_1}+\sqrt{\lambda^{2}-\bar{z}_2}} &\sqrt{1-\lambda^{-2}\bar{z}_2}.
	\end{bmatrix}
\end{align*}
Note that $\bar{z}_1$ and $\bar{z}_2$ are independent variables. We can temporarily use 
\begin{equation*}
	\sqrt{\lambda_1^{2}-\bar{z}_1}  \quad \mbox{and}\quad \sqrt{\lambda_2^{2}-\bar{z}_2},
\end{equation*}
and consider the integral 
\begin{align*}
	&\int_{\R}\d h\,\exp\left( \frac{1}{6}h^3\right)\exp\left( -\frac{1}{2}h\lambda^2\right)\\
	&\exp\left(-\frac{1}{\lambda_1}\frac{\partial}{\partial \lambda_1}\frac{\partial}{\partial s_1}-\frac{1}{\lambda_2}\frac{\partial}{\partial \lambda_2}\frac{\partial}{\partial s_2}\right)\frac{\lambda_1+\lambda_2}{\lambda_1\lambda_2}\exp\left( \frac{1}{6}s_1^3\right)\exp\left( \frac{1}{6}s_2^3\right)\\
	& \frac{1}{2\pi}\int_{\C}[\d s]\,\frac{(h+\lambda_1-s_1)(h+\lambda_2-s_2)-s\bar{s}}{(h+\lambda_1+s_1)(h+\lambda_2+s_2)+s\bar{s}} e^{-\frac{1}{2}(\lambda_1+\lambda_2)s\bar{s}}.
\end{align*}
The complex integral above is related the special function - the exponential integral $Ei(x)$, which has a somewhat complicated formal expansion. After performing the Weierstrass transform on $s_1$ and $s_2$, we can have two variables $y_1$ and $y_2$ like the variable $y$ in Lemma \ref{ytoM+1}. However, the polynomial averaging procedure requires a fraction $(y_1-y_2)(y_1+y_2)^{-1}$ in the product part like Eq.\eqref{prod}, and there seems to be no way to obtain this term.

\section*{Acknowledgement}
The author is grateful to Alexander Alexandrov for the insightful discussions and guidance.

\section*{Appendix}
\subsection*{A.1 The Gaussian integral over complex numbers}
In the complex plane, we write
\begin{equation*}
	z=x+y\i, \quad \bar{z}=x-y\i.
\end{equation*}
Then, the differentials are expressed as
\begin{equation*}
	\d z=\d x+\i\,\d y, \quad \d\bar{z}=\d x-\i\,\d y.
\end{equation*}
We know that $x$ and $y$ are real-valued and independent to each other. However, the complex variables $z$ and $\bar{z}$ are related, hence they can not be considered as independent variables. Using the exterior product (wedge product), we have
\begin{equation*}
	\d \bar{z}\wedge \d z = -\, \d z \wedge\d\bar{z}=2\i \,\d x\wedge\d y.
\end{equation*}
Note that, as the area element $\d x\d y$ in the double integrals, 
\begin{equation*}
	\d x\d y=\d x\wedge\d y= \frac{\d\bar{z}\wedge \d z }{2\i}=[\d z].
\end{equation*}

The standard Gaussian integral over real line $\R$ is in the form:
\begin{equation*}
	\int_{-\infty}^{\infty}\d \mu(x)=1,\quad\mbox{where}\quad \d \mu(x)=\frac{1}{\sqrt{2\pi}}e^{\frac{x^2}{2}} \,\d x
\end{equation*}
is the Gaussian measure. Sometimes we also call the integrals over the Gaussian measure the Gaussian integrals. Furthermore, we refer to the following integral as the Gaussian integrals over complex numbers,
\begin{equation*}
	\int_{\C} f(z,\bar{z})e^{-\frac{1}{2}z\bar{z}}\, [\d z].
\end{equation*}
The integral is an area integral over the complex plane. The notation $f(z,\bar{z})$ means that $f$ is a non-holomorphic function in general. 

Let $u$ be a formal parameter, and $F(uz,\bar{z})$ be a formal power series. Then,
\begin{equation}\label{equal}
	\frac{1}{2\pi}\int_{\C}F(uz,\bar{z})e^{\frac{1}{2}s\bar{z}}e^{-\frac{1}{2}z\bar{z}}\, [\d z]=\left.\left( e^{2\frac{\partial^2}{\partial s\partial s_{-}}}\cdot F(us,s_{-})\right)\right\vert_{s_{-}=0}
\end{equation}
The parameter $u$ is to make sure that the formal integral is well-defined. To prove this equality, let $m$ and $n$ be non-negative integers, and $z=re^{\i\theta}$. Using polar coordinates $[\d z]=\d x\d y= r\,\d r\d \theta$, one can show that
\begin{align*}
	&\int_{\C} \bar{z}^mz^n e^{-\frac{1}{2}\bar{z}z}\,[\d z] \nonumber \\
	=&\int_0^{\infty}r^{m+n+1}e^{-\frac{1}{2}r^2}\d r\int_0^{2\pi} e^{\i(n-m)\theta}\d \theta\nonumber\\
	=&2\pi\delta_{mn}\int_0^{\infty}(\sqrt{2u})^{m+n}e^{-u}\d u,
\end{align*}
where $\delta_{mn}=0$ for $m\neq n$, and $\delta_{mn}=1$ for $m=n$. If $m=n$, the above integral is equal to $(2\pi) 2^nn!$. Now, for $m>n$, we have
\begin{equation*}
	\frac{1}{2\pi}\int_{\C}\bar{z}^mz^ne^{\frac{1}{2}s\bar{z}}e^{-\frac{1}{2}z\bar{z}}\, [\d z]=\left.\left( e^{2\frac{\partial^2}{\partial s\partial s_{-}}}\cdot s_{-}^ms^n\right)\right\vert_{s_{-}=0}=0.
\end{equation*}
When $m\leq n$, we have
\begin{equation*}
	\frac{1}{2\pi}\int_{\C}\bar{z}^mz^ne^{\frac{1}{2}s\bar{z}}e^{-\frac{1}{2}z\bar{z}}\, [\d z]
	=\frac{2^mn!}{(n-m)!}s^{n-m},
\end{equation*}
and
\begin{equation*}
	\left.\left( e^{2\frac{\partial^2}{\partial s\partial s_{-}}}\cdot s_{-}^ms^n\right)\right\vert_{s_{-}=0}=\frac{2^m}{m!}\left(\frac{\partial^2}{\partial s\partial s_{-}} \right)^m\cdot s_{-}^ms^n=\frac{2^mn!}{(n-m)!}s^{n-m}.
\end{equation*}
This shows that Eq.\eqref{equal} holds for any formal power series.

\subsection*{A.2 The Weierstrass transform}
We consider the {\bf heat equation} (we refer to \cite{W} for more details about the heat equation):
\begin{equation*}
	\frac{\partial^2}{\partial x^2}H=\frac{\partial}{\partial t}H.
\end{equation*}
Let 
\begin{equation*}
	H(x,t)=e^{t\frac{\partial^2}{\partial x^2}}\cdot F(x)=\sum_{k=0}^{\infty}\frac{F^{(2k)}(x)t^k}{k!}, \quad x\in\R, t\geq 0, 
\end{equation*}
where $F(x)=H(x,0)$ is the initial function. The function $H(x,t)$ is in fact the convolution of $H(x,0)$ as the following integral (also known as the {\bf Weierstrass transform}), 
\begin{equation}\label{WT}
	H(x,t)=\frac{1}{2\sqrt{\pi t}}\int_{-\infty}^{\infty}\exp\left( -\frac{(x-y)^2}{4t}\right)F(y) \,\d y.
\end{equation}

Let us also consider the formal Weierstrass transform. Assume $s\in \R$ and $t>0$.  By the definition \eqref{WT}, on monomials, we have
\begin{align*}
	e^{t\frac{\partial^2}{\partial s^2}}\cdot s^n&=\frac{1}{2\sqrt{\pi t}}\int_{-\infty}^{\infty}\exp\left( -\frac{x^2}{4t}\right)(x+s)^n \d x\\
	=&\frac{1}{\sqrt{2\pi }}\int_{-\infty}^{\infty}\exp\left( -\frac{x^2}{2}\right)(\sqrt{2t}x+s)^n \d x.
\end{align*}
Furthermore, for a power series $F(s)$,
\begin{align*}
	e^{t\frac{\partial^2}{\partial s^2}}\cdot F(s)
	&=\frac{1}{\sqrt{2\pi }}\int_{-\infty}^{\infty}\exp\left( -\frac{x^2}{2}\right)F(\sqrt{2t}x+s)\, \d x.
\end{align*}
We can define the inverse Weierstrass transform to be
\begin{align}
	e^{-t\frac{\partial^2}{\partial s^2}}\cdot F(s)&=\frac{1}{2\sqrt{\pi t}}\int_{-\infty}^{\infty}\exp\left( -\frac{y^2}{4t}\right)F(\i y+s)\,  \d y\nonumber\\
	&=\frac{1}{2\sqrt{\pi t}}\int_{-\infty}^{\infty}\exp\left( -\frac{y^2}{4t}\right)F(-\i y+s)\,  \d y\label{invW}.
\end{align}
In this way, we have
\begin{align*}
	e^{-t\frac{\partial^2}{\partial s^2}}e^{t\frac{\partial^2}{\partial s^2}}\cdot F(s)
	=&\frac{1}{2\pi}\int_{-\infty}^{\infty}\int_{-\infty}^{\infty}\exp\left( -\frac{x^2+y^2}{2}\right)F(\sqrt{2t}x+\sqrt{-2t}y+s)\d x\d y\\
	=&\frac{1}{2\pi}\int_{\C} e^{-\frac{1}{2}z\bar{z}}F(\sqrt{2}z+s)\, [\d z]\\
	=& F(s).
\end{align*}
We can also deduce the above equality using the  delta function as
\begin{align*}
	e^{-t\frac{\partial^2}{\partial s^2}}e^{t\frac{\partial^2}{\partial s^2}}\cdot F(s)
	=&\frac{1}{2\sqrt{\pi t}}e^{-t\frac{\partial^2}{\partial s^2}}\int_{-\infty}^{\infty}\exp\left( -\frac{(x-s)^2}{4t}\right)F(x) \,\d x\\
	=&\frac{1}{4\pi t}\int_{-\infty}^{\infty}\int_{-\infty}^{\infty}\exp\left( -\frac{y^2}{4t}\right)\exp\left( -\frac{(x-s-\i y)^2}{4t}\right)F(x)\,\d x\d y\\
	=&\frac{1}{4\pi t}\int_{-\infty}^{\infty}\exp\left( -\frac{(x-s)^2}{4t}\right)F(x)\,\d x\int_{-\infty}^{\infty}\exp\left( \i y\frac{x-s}{2t}\right)\d y\\
	=&\frac{1}{2 t}\int_{-\infty}^{\infty}\delta\left(\frac{x-s}{2t}\right)\exp\left( -\frac{(x-s)^2}{4t}\right)F(x)\,\d x\\
	=& \int_{-\infty}^{\infty}\delta\left(x\right)\exp\left( -tx\right)F(2tx+s)\,\d x\\
	=&F(s).
\end{align*}
Here we used the standard Fourier transform
\begin{equation*}
	F(y)=\frac{1}{\sqrt{2\pi}}\int_{-\infty}^{\infty} f(x)e^{\i xy}\d x.
\end{equation*} 
If $f(x)=1$, then 
\begin{equation*}
	\int_{-\infty}^{\infty} e^{\i xy}\d x=2\pi\delta(y),
\end{equation*}
where $\delta$ is the Dirac's delta function with the fundamental property 
\begin{equation*}
	\int_{-\infty}^{\infty} g(y)\delta(y-a)\d y=g(a).
\end{equation*}

	\vspace{10pt} \noindent
	\\
	\footnotesize{\sc gehao wang }\\
	Department of Mathematics,\\
	College of Information Science and Technology/College of Cyberspace Security,\\
	Jinan University, Guangzhou, China. \\
	\footnotesize{E-mail address:  gehao\_wang@hotmail.com}
	
\end{document}